\documentclass[reqno,xcolor=dvipsnames]{amsart}
\usepackage[dvipsnames]{xcolor}

\usepackage{geometry}
\usepackage[round,comma]{natbib}                %

\usepackage{enumerate}
\usepackage{mathabx}

\usepackage{amsmath}
\usepackage{amsthm}
\usepackage{amsfonts}
\usepackage{amssymb}
\usepackage{dsfont}

\usepackage{nicefrac}
\usepackage{booktabs}
\usepackage{mathrsfs}
\usepackage{mathtools}                          %

\newcommand{\ccF}{{\mathscr F}}

\newcommand{\Ind}{{\mathds 1}}
\newcommand{\ind}[1]{\Ind_{\{#1\}}}

\newcommand{\restr}{\mathbf{\kern0.3ex%
 \vert\kern-0.3ex}\backprime\kern0.3ex}

\newtheorem{theorem}{Theorem}[section]
\newtheorem{lemma}[theorem]{Lemma}              %

\theoremstyle{definition}
\newtheorem{assumption}[theorem]{Assumption}%

\usepackage[backgroundcolor=white,bordercolor=orange]{todonotes}

\usepackage[colorlinks,urlcolor=red,citecolor=blue,linkcolor=red]{hyperref}

\usepackage{MnSymbol}  %

 \newcommand{\Kom}[1]{}

\newcommand{\cadlag}{c\`{a}dl\`{a}g}

\newcommand{\R}{\mathbb{R}}

\newcommand{\beq}{\begin{equation}}
\newcommand{\eeq}{\end{equation}}

\newcommand{\bbF}{\mathbb{F}}

\newcommand{\Q}{\mathbb{Q}}

\newcommand{\dbra}[1]{[\kern-0.15em[ #1 ]\kern-0.15em]} 

\begin{document}

\renewcommand{\theenumi}{\roman{enumi}}

\title{Defaultable term structures driven by semimartingales}
	    \author[Gümbel]{Sandrine Gümbel}
		\author[Schmidt]{Thorsten Schmidt}
		\address{Albert-Ludwigs University of Freiburg, Ernst-Zermelo-Str. 1, 79104 Freiburg, Germany.}
		\email{sandrine.guembel@stochastik.uni-freiburg.de}
		\email{thorsten.schmidt@stochastik.uni-freiburg.de}
    
    \date{\today. }
    \thanks{ Financial support from the German Research Foundation (DFG) within project No.~SCHM 2160/9-1 is gratefully acknowledged.}

\begin{abstract} We consider a market with a term structure of credit risky bonds in the single-name case. We aim at minimal assumptions extending existing results in this direction:
first, the  random field of forward rates is driven by a general semimartingale. Second, the Heath-Jarrow-Morton approach is extended with an additional component capturing those future jumps in the term structure which are visible from the current time. Third, the associated recovery scheme is as general as possible, it is only assumed to be non-increasing. 
In this general setting we derive generalized drift conditions which characterize when a given measure is a local martingale measure, thus yielding no asymptotic free lunch with vanishing risk (NAFLVR), the right notion for this large financial market to be free of arbitrage.  
 \\%

\end{abstract}


\maketitle

\section{Introduction}\label{cha:creditRisk}
The risk that a counterparty of a financial contract is not able to fulfil its obligations, in other words if it {\em defaults}, is in banking  known as credit risk. The default of a corporate or sovereign entity might be due to a variety of reasons, e.g. bankruptcy or failure to pay due to deteriorating business  conditions. In this work we aim at a general framework for a single-name credit market subject to default risk. In particular, we  derive conditions to ensure absence of arbitrage in an extended \cite{HJM} (HJM) model allowing for a non-absolutely continuous term structure. %

The two main modeling approaches in credit risk are  {\em structural models} also known as firm-value models and {\em reduced-form models}, often also called {\em intensity based approaches}.
The structural approach was introduced by \cite{merton1974pricing}. Credit events appear in correspondence to the firm's value relative to some default triggering barrier. The firm's value is described by a stochastic process. One advantage of the structural approach consists in the economically intuitive picture of the direct connection of default events and the firm's capital structure. One drawback is that the determination of the firm's capital structure constitutes a challenging task. 

The reduced-form approach was introduced by \cite{jarrow1995pricing} and \cite{artzner1995default}.
In this model class the time of default is modeled by an exogenous random variable. The default cannot be predicted and may occur at any time. Typically, one assumes the existence of a default intensity. 
Extended versions of reduced-form approaches drop the assumption of the existence  of a default intensity. This allows to incorporate default at fixed and predictable times with positive probability. In that sense they unify the reduced-form and structural approach.

The literature on defaultable term structure modelling is immense and we refer to \cite{Ammann2001}, \cite{Schonbucher}, \cite{Bielecki2004}, \cite{SchmidtStute2004}, \cite{duffie2012credit}  and \cite{doumpos2019analytical} for an overview. %
Following the seminal paper  \cite{merton1974pricing}, early structural models are treated in \cite{black1973pricing}, \cite{black1976valuing}, \cite{geske1977valuation} and \cite{geske1984valuation}. While the driving processes to model the evolution of the firm's value in structural models are commonly diffusion processes, \cite{zhou1997jump} proposes a jump-diffusion process as driving process and values default-risky securities. 
\cite{duffie2001term} and \cite{FreySchmidt2009} present  structural models that are consistent with a reduced-form approach by considering incomplete information. \cite{cetin2004modeling} obtain a reduced-form approach from their structural credit risk model by considering a reduced manager information set.
In contrast to these structural approaches, the reduced-form setting starts with  \cite{jarrow1995pricing} and \cite{artzner1995default}, followed for example by  \cite{duffie1996recursive}, \cite{jarrow1997markov}, \cite{lando1998cox}, \cite{madan1998pricing} and \cite{elliott2000models}, amongst many other works. Jumps are incorporated in frameworks such as
\cite{Bjoerk1997}, 
 \cite{eberlein2003defaultable}, \cite{ozkan2005credit}, \cite{eberlein2005levy}, \cite{filipovic2008existence},   \cite{eberlein2013rating} or \cite{Cuchiero2016}. %

The approach which we study here is general enough to encompass both approaches. It appears that 
\cite{belanger2004general} is one of the earliest works where this is suggested. Moreover, in discrete time, discontinuities arise naturally in the term structure which was already discussed in \cite{filipovic2002markovian}. A more general approach was proposed in \cite{GehmlichSchmidt2016MF}. Stochastic discontinuities were covered with random measures for the first time in \cite{fontana2018general} and here we extend these approaches further by allowing for semimartingale processes as drivers for the random fields of forward rates. This significantly increases the amount of technicality of the framework, but reveals the presence of additional terms appearing in the conditions that ensure no-arbitrage. For discontinuities in a setting with two filtrations we refer to \cite{jiao2015generalized} and \cite{jiao2018modeling}.

The paper is structured as follows: In Section \ref{sec:cRdefaultablemodeling} we introduce the basic framework for our setting. We continue by studying the simpler setting with vanishing recovery in Section \ref{MainResults}. The main theorem is contained in Section \ref{sec:cRrecovery}, where we allow for a general recovery process. Section \ref{sec:conclusion} concludes.

\section{General defaultable term structure modeling}\label{sec:cRdefaultablemodeling}
In this section we provide a general modeling framework for credit risky bonds in an extended reduced-form HJM-framework featuring a non-absolutely continuous term structure of credit risky bonds. Our main goal is to propose a general semimartingale framework. Since semimartingales allow for  stochastic discontinuities, such an approach necessarily requires an extension of the classical HJM-framework.

We consider a credit risky financial market which contains credit risky bonds for all available maturities. Such a market containing uncountably many traded products  is a large financial market and can be tackled with techniques from this strand of literature. This guides us to the economic notion  of no-arbitrage for large financial markets, no asymptotic free lunch with vanishing risk (NAFLVR). We use the techniques for term-structure markets  developed  in  \cite{KleinSchmidtTeichmann2015} and \cite{CuchieroKleinTeichmann}.
A direct application of this general setting (compare \cite{fontana2020term}) yields as a sufficient condition of NAFLVR the existence of an equivalent local martingale measure (ELMM).  In the following, we therefore derive necessary and sufficient conditions for a reference probability measure to be a local martingale measure. 

\subsection{The market of credit risky bonds}
We consider an infinite time horizon, while models with a finite time horizon $\mathbb{T}<+\infty$ can be encompassed by stopping the relevant processes at $\mathbb{T}$.
We assume that the stochastic basis $(\Omega,\ccF,\bbF,\Q)$ is rich enough for the following processes to exist. 
We use the convention $(t,t]=\emptyset$, for all $t \in \R_+$. 
We refer to \cite{JacodShiryaev} for  the essential tools of  stochastic analysis which we require here.

Discounting is done with respect to the bank account 
\begin{equation}\label{eq:numeraire}
	B =\exp\left(\int_0^\cdot r_s ds\right),
\end{equation}
with a progressive process $r = (r_t)_{t\geq 0}$ (the short rate) satisfying $\int_0^T|r_s|ds < + \infty$ a.s. for all $T>0$. For practical applications one typically chooses the OIS rate to construct such a num\'{e}raire. We remark that $B$ is a strictly positive, adapted and absolutely continuous process. This assumption could  be relaxed, i.e.~$B$ replaced by a general, positive semimartingale at the cost of more involved formulas. 

For $0 \leq t \leq T$ we denote by $P(t,T)$ the price at date $t$ of a credit risky bond with maturity $T\in \R_+$. 
We assume that the credit risky bond price admits the following generalized HJM-representation
\begin{equation}\label{termstructure}
P(t,T) = \xi_t \exp\left(-\int _t^T f(t,u)du - \int_{(t,T]} g(t,u) \mu_t(du) \right), 
\end{equation}
where $(f(t,T))_{0 \leq t \leq T}$ and $(g(t,T))_{0 \leq t \leq T}$ are so-called \emph{forward rate processes}. The integral with respect to $\mu_t(du)$ encodes all the information received up to date $t$ about possible future risky dates where default might occur with positive probability. 
The process $\xi$ allows for incorporating general recovery schemes.

The  HJM formulation, see \cite{HJM}, assumes absolute continuity and is covered as special case with $g \equiv 0$. Here, we  extended it in two ways: on the one hand by assuming the forward rates to be random fields driven by general semimartingales. And on the other hand by introducing  the  measure-valued process $(\mu_t(du))_{t\geq 0}$ encompassing possible singular and jumps part, following \cite{fontana2018general}.

The measure-valued process $\mu(du) = (\mu_t(du))_{t\geq 0}$ is derived from an optional non-negative random measure on $\Omega \times \mathbb{R}_+ \times E$, with $(E, \mathcal{B}_E)$ denoting a Polish space with its Borel sigma-field in the following way: 
$$
	\mu_t(du) :=\mu([0,t]\times du).
$$
Furthermore, we make  the following assumptions.
\begin{assumption}\label{ramdommeasure}
	The random measure $\mu(\omega; ds,du)$ is a non-negative optional random measure on $\Omega \times\R_+\times E$ satisfying
	\begin{enumerate}[(i)]
		\item $\mu(\omega;ds,du) = \Ind_{\{ s < u\}} \mu(\omega;ds,du)$, for all $(s,u) \in [0,T]\times [0,T], T \in \mathbb{R}_+$ and $\omega \in \Omega;$
		\item there exists a sequence $(\sigma_n)_{n \in \mathbb{N}}$ of stopping times increasing almost surely to infinity such that $\mathbb{E}^\Q[\mu_{\sigma_n}([0,T])] < \infty$ for every $n \in \mathbb{N}$ and $T \in \mathbb{R}_+$.
	\end{enumerate}
\end{assumption}
 The first assumption corresponds to the situation that new information arriving at time $s$ only affects the likelihood of default of the future and not of the past. The second assumption is an integrability condition to ensure that the measure $\mu(ds,du)$ is predictably $\sigma$-finite %
 and that the random variable $\mu_t([0,T])$ is almost surely finite for all $0 \leq t\leq T$ and $T\in \R_+$.

Furthermore, we define for every $T \in \mathbb{R}_+$ the process $\bar{\mu}^{(T)} = (\bar{\mu}_t^{(T)})_{0 \leq t\leq T}$ by
\begin{equation}\label{eq:defmubar}
\bar{\mu}_t^{(T)} := \mu([0,T]\times [0,t]), \quad \text{for all } t \in [0,T],
\end{equation}
measuring the effect of risky dates in the period $[0,t]$ on the basis of all available information over the time interval $[0,T]$. Assumption \ref{ramdommeasure} ensures that the process $\bar{\mu}^{(T)}$ is predictable and increasing and admits a decomposition in an absolutely continuous part, a singular continuous part and a jump part.

\subsection{The forward rate processes}
For every $T\in \R_+$, the \emph{generalized forward rate processes} appearing in \eqref{termstructure} are assumed to be general semimartingale random fields of the form
\begin{align}
f(t,T) & = f(0,T) + \int_0^t  a(s,T) dA_s + \int_0^t b(s,T) dX_s, \label{tsf}\\
g(t,T) & = g(0,T) + \int_0^t  \alpha(s,T) dA_s + \int_0^t \beta(s,T) dX_s, \label{tsg}
\end{align}
for all $0 \leq t \leq T$, where $X$ is a $\R^d$-valued semimartingale and $A$ a $\R$-valued process of finite variation. Without loss of generality we may assume that $A$ is increasing.

We impose the following assumptions in order to ensure that all integrals in \eqref{tsf} -- \eqref{tsg} are well-defined and in order to apply suitable versions of (stochastic) Fubini theorems in the sequel.

\begin{assumption}\label{AssumptionFubiniandEx}
	The following conditions hold a.\,s.:
	\begin{enumerate}[(i)]
		\item the {\em initial forward curves} $T \mapsto f(0,T)$ and $T \mapsto g(0,T)$ are real-valued, $\mathcal{F}_0 \otimes \mathcal{B}(\R_+)$-measurable, and satisfy $\int_0^T |f(0,u)| du < + \infty$ and $\int_0^T |g(0,u)| \mu_t(du) < + \infty$ for all $t\in [0,T]$ and $T\in\R_+$;
		\item the \emph{drift processes} $a(\cdot\, ; \,\cdot,\cdot): \Omega \times \R_+ \times \R_+ \rightarrow \R$  and $\alpha(\cdot\, ;\, \cdot, \cdot): \Omega \times \R_+ \times \R_+ \rightarrow \R$ are real-valued and $\mathcal{O}\otimes \mathcal{B}(\mathbb{R}_+)$-measurable, for every $t \in \R_+$, where $\mathcal{O}$ denotes the optional $\sigma$-field. They satisfy $a(\omega; t,T) = 0$ and $\alpha(\omega; t,T) = 0$ for all $0 \leq T < t < +\infty$, and 
		\begin{equation*}
			\int_0^T\int_0^T |a(s,u)| du d|A|_s < +\infty \quad \text{for all } T \in \mathbb{R}_+,
		\end{equation*}
		\begin{equation*}
			\int_0^T \int_0^T |\alpha(s,u)|d|A|_s \mu_t(du) < +\infty \quad \text{for all } 0 \leq t \leq T,\ T \in \mathbb{R}_+,
		\end{equation*}
		where $d|A|_s$ denotes the Stieltjes measure induced by the total variation $|A|$ of $A$;
		\item the \emph{volatility processes} $b(\cdot \,;\, \cdot,\cdot): \Omega \times \R_+ \times \R_+ \rightarrow \R^d$ and $\beta(\cdot \, ; \, \cdot,\cdot): \Omega \times \R_+ \times \R_+ \rightarrow \R^d$ are $\R^d$-valued and $\mathcal{P}\otimes\mathcal{B}(\mathbb{R}_+)$-measurable, for every $t\in \R_+$, where $\mathcal{P}$ denotes the predictable $\sigma$-field. They satisfy  $b(\omega; t,T) = 0$ and $\beta(\omega; t,T) = 0$ for all $0 \leq T < t < +\infty$, and 
		\begin{equation*}
			\left(\left(\int_0^T |b^i(s,u)|^2 du\right)^{\frac{1}{2}}\right)_{0\leq s \leq T} \in L(X^i),
		\end{equation*}
		\begin{equation*}
			\left(\left(\int_0^T |\beta^i(s,u)|^2 \mu_s(du)\right)^{\frac{1}{2}}\right)_{0\leq s \leq T} \in L(X^i),
		\end{equation*}
		for all $T \in \mathbb{R}_+$, and $i = 1,\dots, d$, where $L(X^i)$ denotes the set of processes which are integrable with respect to $X^i$.
	\end{enumerate}
\end{assumption}
We note that the second integrability condition from Assumption \ref{AssumptionFubiniandEx} (ii) and (iii) ensures the integrability of some integrals in Lemma \ref{eq:reprcRBond}. 

We impose in addition the following assumptions in order to ensure the interchangeability of some integrals in the following calculations.  
\begin{assumption}\label{ass:lackingFubini}
	The following conditions hold a.s.:
	\begin{enumerate}[(i)]
		\item for all $s \leq t \leq T$ and $T\in \R_+$ it holds that
		\begin{align*}
		& \int_0^t \int_0^T  \int_s^t \Ind_{[0,u)}(v-) \alpha(v,u)dA_v \mu(ds,du)\\
		&\quad =  \int_0^t \int_0^v  \int_0^T \Ind_{[0,u)}(v-) \alpha(v,u) \mu(ds,du)dA_v;
		\end{align*}
		\item 	furthermore, for all $s \leq t \leq T$ and $T\in \R_+$ it holds that
		\begin{align*}
			 &\int_0^t \int_0^T \int_s^t \Ind_{[0,u)}(v-)\beta(v,u)dX_v \mu(ds,du)\\
			&\quad  = \int_0^t \int_0^v \int_0^T \Ind_{[0,u)}(v-)\beta(v,u) \mu(ds,du) dX_v.
		\end{align*}
	\end{enumerate}
\end{assumption}
We set for all $0 \leq t \leq T$ and $T \in \R_+$, 
\begin{align}\label{eq:barprocesses}
\bar{a} (t,T) & := \int_t^T a(t,u) du,\nonumber\\
\bar{b} (t,T) & := \int_t^T b(t,u) du,\nonumber\\
\bar{\alpha} (t,T) & := \int_t^T \alpha(t,u) \mu_t(du),\nonumber\\
\bar{\beta} (t,T) & := \int_t^T \beta(t,u) \mu_t(du), 
\end{align}
and 
\begin{align}\label{eq:barjoin}
\bar{A} (t,T) & :=   \bar{a} (t,T) + \bar{\alpha} (t,T) =  \int_t^T a(t,u) du +  \int_t^T \alpha(t,u) \mu_t(du),\nonumber\\
\bar{B} (t,T) & :=   \bar{b} (t,T) + \bar{\beta} (t,T) = \int_t^T b(t,u) du + \int_t^T \beta(t,u) \mu_t(du).
\end{align}

It remains to specify the process $ \xi = (\xi_t)_{t\geq 0}$ from the bond price specification in Equation \eqref{termstructure}, controlling the recovery property. 
We distinguish two cases. First, we consider the simpler case  of zero recovery  in Section \ref{MainResults}. The general case is treated in the following  Section \ref{sec:cRrecovery}.  

\section{Defaultable term structure modeling with zero recovery}\label{MainResults}

Default recovery is the first step in studying credit risky markets. Here the assumption is that at a general stopping time $\tau$, the default time, the company defaults
and the credit risky bond becomes worthless. Combining this with a recovery which is no longer subject to default, i.e.~relying on the interest rate market, leads to the so-called class of recovery of face-value approaches. Here we consider a fully general approach for the stopping time. 

The only assumption we make is that the default time $\tau$ is an $\mathbb{F}$-stopping time. Since the default of a company is public information this is very reasonable.  We define the associated \emph{default indicator process} $H = (H_t)_{0 \leq t\leq T}$ by 
\begin{equation*}
H_t = \Ind_{\{\tau \leq t\}} , \quad \text{ for } 0 \leq t \leq T \text{ and }T \in \mathbb{R}_+. 
\end{equation*}
Since we are working under zero recovery in this section, we assume 
\begin{align} \label{ass:defwithoutrec}
	\xi = (1-H). 
\end{align}
For the following we introduce the following notation: 
\begin{equation}\label{termstructurezerorec}
P(t,T) = (1-H_t)F(t,T)G(t,T),
\end{equation}
where the random fields $F$ and $G$ are induced by the forward rates $f$ and $g$,
\begin{align*}
F(t,T) &:= \exp\left(-\int_t^Tf(t,u)du\right) 
\intertext{ and }  
G(t,T) &:= \exp\left(-\int_{(t,T]}g(t,u)\mu_t(du)\right). 
\end{align*}

A first result provides an alternative representation of credit risky bond prices $P(t,T)$. %

\begin{lemma}\label{eq:reprcRBond}
	Suppose that Assumptions \ref{ramdommeasure} -- \ref{ass:lackingFubini} hold. Under zero recovery, and for every $T>0$, the process $(P(t,T))_{0\leq t \leq T}$ admits the representation 
	$$
	P(t,T) = (1-H_t)\exp\bigl(X_t^{(T)}\bigr), \quad \text{for all } 0 \leq t \leq T,
	$$
	where the semimartingale $X^{(T)} = (X_t^{(T)})_{0 \leq t \leq T}$ is defined by 
	\begin{align}\label{processXT}
	X_t^{(T)} & :=  \int_0^t f(s,s) ds - \int_0^T f(0,u) du - \int_0^t \bar{A}(s,T) dA_s - \int_0^t \bar{B}(s,T) dX_s\nonumber\\
	&   \quad - \int_0^t \int_{(s,T]} g(s,u)  \mu(ds,du) + \int_0^t g(s-,s) d\bar{\mu}_s^{(T)}.
	\end{align}
\end{lemma}

\begin{proof}%
	For every $T \in \R_+$ we define the process $X^{(T)} = (X_t^{(T)})_{0 \leq t \leq T}$ as
	\begin{equation}\label{eq:temp}
	X_t^{(T)}  :=  \log(P(t,T)) - \log(1-H_t) = \log(F(t,T)) + \log(G(t,T)).
	\end{equation}
	In a first step we calculate an alternative representation of $\log(F(t,T))$. For this purpose, note that%

	\begin{align*}
	\int_t^T f(t,u) du 
	 		    & =  \int_0^T f(0,u) du + \int_0^t \bar{a}(s,T) dA_s + \int_0^t \bar{b}(s,T) dX_s-\int_0^t f(u,u) du,
	\end{align*}
	following \cite{HJM} with the slight generalization replacing the Lesbesgue measure in the second integral by $dA$. The application of the 
	stochastic Fubini Theorem (\citet{Protter}, Theorem IV.65) is justified due to Assumption \ref{AssumptionFubiniandEx}.

	In the next step, we derive a representation of $G(t,T)$. 
	Decompose  the semimartingale $X = M + \tilde A$ into a local martingale $M$ and a process of finite variation $\tilde A$. The integral $\int_0^\cdot \bar{\beta}(s,T)dX_s$ is well-defined for every $T\in\mathbb{R}_+$ due to Hölder's inequality, Assumption \ref{ramdommeasure} (ii), and Assumption \ref{AssumptionFubiniandEx} (iii). Indeed, it holds for every $i = 1,...,d$
	\begin{align*}
		\int_0^T \bigl( \bar\beta^i(s,T)\bigr)^2 d\langle M^i \rangle_s 
		& = \int_0^T \left(\int_s^T \beta^i(s,u) \mu_s(du)\right)^2 d\langle M^i \rangle_s\\
		& \leq \mu_T([0,T]) \int_0^T \int_0^T |\beta^i(s,u)|^2 \mu_s(du)d\langle M^i\rangle_s < + \infty.
	\end{align*}
	Due to the definition of the process  $\mu(du) = (\mu_t(du))_{t \geq 0}$ it holds that
	\begin{align}
	- \log G(t,T) & =  \int_{(t,T]} g(t,u) \mu_t(du) = \int_0^t \int_{(t,T]} g(t,u) \mu(ds,du) \nonumber\\
	& =   \int_0^t \int_0^T \Ind_{\{u > t\}} g(t,u) \mu(ds,du) \label{logG}.
	\end{align}
	Integration by parts  and Equation \eqref{tsg} yield
	\begin{align}
	\Ind_{[0,u)}(t)g(t,u) & =  g(0,u) + \int_0^t \Ind_{[0,u)}(v-)dg(v,u) \nonumber\\
	& \quad  + \int_0^t g(v-,u)d(\Ind_{[0,u)}(v)) 
	+ [\Ind_{[0,u)}(\cdot ),g(\cdot,u)]_t\nonumber\\
	& =  g(0,u) + \int_0^t \Ind_{[0,u)}(v-) \alpha (v,u)dA_v \nonumber\\
	&\quad + \int_0^t \Ind_{[0,u)}(v-)\beta(v,u)dX_v 
	 -  g(u-,u) \Ind_{\{u \leq t \}} 
	  + [\Ind_{[0,u)}(\cdot),g(\cdot,u)]_t,\label{indg}
	\end{align}
	where both integrals are well-defined due to Assumption \ref{AssumptionFubiniandEx}. %
	Now combine \eqref{logG} and \eqref{indg} which leads to
	\begin{align}
	\int_{(t,T]} g(t,u) \mu_t(du) & =  \int_0^t \int_0^T g(0,u) \mu(ds,du) \nonumber \\
	&   \quad + \int_0^t \int_0^T  \int_0^t \Ind_{[0,u)}(v-) \alpha(v,u)dA_v \mu(ds,du) \nonumber \\
	&   \quad + \int_0^t \int_0^T \int_0^t \Ind_{[0,u)}(v-)\beta(v,u)dX_v \mu(ds,du) \nonumber \\
	&  \quad - \int_0^t \int_0^T g(u-,u) \Ind_{\{u \leq t \}} \mu(ds,du)\nonumber\\
	&   \quad + \int_0^t \int_0^T  [\Ind_{[0,u)}(\cdot),g(\cdot,u)]_t \mu(ds,du)\nonumber\\
	& =:  (i) + (ii) + (iii) + (iv) + (v). \label{divisionminuslogG}
	\end{align}
	We study the terms separately. Regarding (ii), note that $ \int_0^t \Ind_{[0,u)}(v-) \alpha(v,u)dA_v =  \int_{(0,t]} \Ind_{[0,u)}(v-) \alpha(v,u)dA_v$. 
	By means of \eqref{eq:barprocesses}, %
	\begin{align}
	(ii) & =   \int_0^t \int_0^T  \int_0^s \Ind_{[0,u)}(v-) \alpha(v,u)dA_v \mu(ds,du) \nonumber \\
	&   \quad + \int_0^t \int_0^T  \int_s^t \Ind_{[0,u)}(v-) \alpha(v,u)dA_v \mu(ds,du)\nonumber\\
	& =   \int_0^t \int_0^T  \int_0^s \Ind_{[0,u)}(v-) \alpha(v,u)dA_v \mu(ds,du) 
	    + \int_0^t \int_v^T \alpha(v,u) \mu_v(du) dA_v\nonumber\\
	& =   \int_0^t \int_0^T  \int_0^s \Ind_{[0,u)}(v-) \alpha(v,u)dA_v \mu(ds,du) 
	    + \int_0^t \bar{\alpha}(v,T) dA_v. 
	\end{align}
	 Similarly, 
	\begin{align}
	(iii) 
	& =  \int_0^t \int_0^T \int_0^s \Ind_{[0,u)}(v-)\beta(v,u)dX_v \mu(ds,du)  
	  + \int_0^t  \bar{\beta}(v,T) dX_v.\nonumber
	\end{align}
	We rearrange \eqref{indg} to obtain
	\begin{multline}\label{tosummarize2and3}
	\int_0^s \Ind_{[0,u)}(v-) \alpha (v,u)dA_v + \int_0^s \Ind_{[0,u)}(v-)\beta(v,u)dX_v + [\Ind_{[0,u)}(\cdot),g(\cdot,u)]_s \\
	=  \Ind_{[0,u)}(s)g(s,u) - g(0,u) + g(u-,u) \Ind_{\{u \leq s \}}.
	\end{multline}
	Summarizing the previous results, we obtain for \eqref{divisionminuslogG}
	\begin{align}\label{prevresults}
	\int_{(t,T]} g(t,u) \mu_t(du) & =  \int_0^t \int_0^T g(0,u) \mu(ds,du) + \int_0^t \bar{\alpha}(v,T) dA_v \nonumber\\
	&   \quad + \int_0^t  \bar{\beta}(v,T) dX_v
	+ \int_0^t \int_0^T \Ind_{[0,u)}(s)g(s,u)  \mu(ds,du)\nonumber\\
	&   \quad - \int_0^t \int_0^T g(0,u)\mu(ds,du)
	 + \int_0^t \int_0^T g(u-,u) \Ind_{\{u \leq s \}}\mu(ds,du)\nonumber\\
	&   \quad - \int_0^t \int_0^T g(u-,u) \Ind_{\{u \leq t \}} \mu(ds,du) \nonumber\\
	& =   \int_0^t \bar{\alpha}(v,T) dA_v 
	+ \int_0^t  \bar{\beta}(v,T) dX_v 
	 + \int_0^t \int_0^T \Ind_{[0,u)}(s)g(s,u)  \mu(ds,du)\nonumber\\
	&   \quad - \int_0^t \int_0^T \Ind_{(s,t]}(u) g(u-,u) \mu(ds,du).
	\end{align}

	By Assumption \ref{ramdommeasure} (i) and the definition of the process $\bar{\mu}^{(T)}$, we obtain 
	\begin{multline}
	\int_0^t \int_0^T \Ind_{(s,t]}(u) g(u-,u) \mu(ds,du)  =  \int_0^t \int_{(s,t]} g(u-,u) \mu(ds,du)\nonumber\\
	=  \int_0^t \int_0^t g(u-,u) \mu(ds,du) 
	=  \int_0^T \int_0^t g(u-,u) \mu(ds,du) =  \int_0^t g(u-,u) d\bar{\mu}_u^{(T)}.
	\end{multline}

	Finally, combining the calculations, we obtain
	\begin{align*}
	\log G(t,T) & =  -\int_{(t,T]} g(t,u) \mu_t(du)\nonumber \\
	& =  -\int_0^t \bar{\alpha}(s,T) dA_s - \int_0^t  \bar{\beta}(s,T) dX_s - \int_0^t \int_{(s,T]} g(s,u)  \mu(ds,du)
	+ \int_0^t g(u-,u) d\bar{\mu}_u^{(T)},
	\end{align*}
which concludes the proof. 

For the semimartingale property we note that $\bar{\mu}^{(T)}$ is a predictable and increasing process, see \citet[Lemma 2.5]{fontana2018general}. Therefore, the process $\int_0^\cdot g(s-,s)d\bar{\mu}^{(T)}_s$ is a predictable finite variation process. $\int_0^\cdot \int_{(s,T]} g^+(s,u) \mu(ds,du)$ is an optional and increasing process, since $\mu(ds,du)$ is a non-negative optional random measure. Hence, $\int_0^\cdot \int_{(s,T]} g(s,u) \mu(ds,du)$ is of finite variation.
\end{proof}

In a next step we intend to find an alternative representation for the defaultable bond price $P(t,T)$ as a stochastic exponential. %
\begin{lemma} \label{LemmaBondPriceasStochExp}
Suppose that Assumptions \ref{ramdommeasure}\,--\,\ref{ass:lackingFubini} hold. Under zero recovery, the credit risky bond price $P(t,T)$ can be represented as 
	\begin{equation}\label{bondPriceStochasticExponential}
	P(t,T) = \mathcal{E}\left( \tilde{X}^{(T)} - H - [\tilde{X}^{(T)}, H]\right)_t, \quad \text{ for each } 0 \leq t \leq T,
	\end{equation}
	where, for each $T \geq 0$ the process $(\tilde{X}_t^{(T)})_{0 \leq t \leq T}$ is defined as
	\begin{align*} 
		\tilde{X}^{(T)}_t & =  X^{(T)}_t + \frac{1}{2}\int_0^t\bigl(\bar{B}(s,T)\bigr)^2 d\langle X^\textup{c}\rangle_s \nonumber \\
		& + \sum_{0 < s \leq t} \Big( -1 +\bar{A}(s,T) \Delta A_s 
		 + \bar{B}(s,T) \Delta X_s + \int_{(s,T]} g(s,u)  \mu(\{s\}\times du) - g(s-,s) \Delta\bar{\mu}_s^{(T)}  \nonumber\\
		& \phantom{\sum_{0 < s \leq t} \Big(} \quad  + e^{-\bar{A}(s,T) \Delta A_s - \bar{B}(s,T) \Delta X_s- \int_{(s,T]} g(s,u)  \mu(\{s\} \times du) + g(s-,s) \Delta\bar{\mu}_s^{(T)}} \Big).
	\end{align*}
\end{lemma}

\begin{proof}
	The process $(H_t)_{0\leq t \leq T}$ is a jump process with one single jump of the size one, such that it follows by means of the definition of the stochastic exponential, see  \citet[II-8.2]{JacodShiryaev}, and by noting that $\Delta H_t = \ind{\tau = t}$, that
	\begin{equation*}
		1 - H_t = \ind{\tau > t} =  \prod_{0 < s \leq t}(1-\Delta H_s) = e^{H_0-H_t} \prod_{0 < s \leq t}(1-\Delta H_s)e^{\Delta H_s} = \mathcal{E}(-H)_t.
	\end{equation*}
	
	Due to Theorem II-8.10 from \cite{JacodShiryaev} we can transform $\exp(X_t^{(T)})$ from the representation of Lemma \ref{eq:reprcRBond} to a stochastic exponential, $\exp(X_t^{(T)}) = \mathcal{E}(\tilde{X}^{(T)})_t$ with
	\begin{equation}\label{stochlog}
	\tilde{X}^{(T)}_t =X^{(T)}_t  + \frac{1}{2}\big\langle (X^{(T)}_\cdot)^\textup{c},(X^{(T)}_\cdot)^\textup{c}\big\rangle_t + (e^x-1-x)\ast\mu^{X^{(T)}}_t,
	\end{equation} 
	where we denote by $\mu^{X^{(T)}}$ the random jump measure of $X^{(T)}$, in the sense of \citet[Proposition II-1.16]{JacodShiryaev}. Due to Lemma 2.4 in \cite{fontana2018general}, %
	and keeping in mind that the continuous local martingale part of a finite variation process is zero, one observes that the quadratic variation $\langle (X^{(T)})^\textup{c}, (X^{(T)})^\textup{c}\rangle$ simplifies to:
	\begin{equation}
	\big\langle (X^{(T)})^\textup{c}, (X^{(T)})^\textup{c}\big\rangle_t = \int_0^t \bigl(\bar{B}(s,T)\bigr)^2 d \langle X^\textup{c}\rangle_s.
	\end{equation}
	
	As a next intermediate step we calculate  the jumps of $X_t^{(T)}$,
	\begin{align*}
		\Delta X_t^{(T)} & =  -\bar{A}(t,T) \Delta A_t - \bar{B}(t,T) \Delta X_t
		- \int_{(t,T]} g(t,u)  \mu(\{t\}\times du) \\
		& \quad + g(t-,t) \Delta\bar{\mu}_t^{(T)}.
	\end{align*}

	Inserting those calculations in equation \eqref{stochlog} we obtain
	\begin{align*} 
		\tilde{X}^{(T)}_t & =   X^{(T)}_t + \frac{1}{2}\int_0^t\bigl(\bar{B}(s,T)\bigr)^2 d\langle X^\textup{c}\rangle_s + \sum_{0 < s \leq t} \Big( -1 +\bar{A}(s,T)\Delta A_s  \nonumber\\
		& \quad  + \bar{B}(s,T) \Delta X_s   + \int_{(s,T]} g(s,u)  \mu(\{s\} \times du) - g(s-,s) \Delta\bar{\mu}_s^{(T)} \nonumber\\
		& \quad  + e^{-\bar{A}(s,T)\Delta A_s - \bar{B}(s,T)\Delta X_s- \int_{(s,T]} g(s,u)  \mu(\{s\}\times du) + g(s-,s) \Delta\bar{\mu}_s^{(T)}} \Big).
	\end{align*}
	With the previous calculation and Yor's formula, see \citet[II-8.19]{JacodShiryaev}, the defaultable bond price $P(t,T)$ can be written as stochastic exponential
	\begin{align*}
		P(t,T) & =  (1-H_t) \exp\bigl(X_t^{(T)}\bigr) = \mathcal{E}(-H)_t \mathcal{E}\bigl(\tilde{X}^{(T)}\bigr)_t \\
		& =  \mathcal{E}\bigl(\tilde{X}^{(T)} - H - [\tilde{X}^{(T)},H]\bigr)_t,
	\end{align*}
	which completes the proof.
\end{proof}

As preparation for the next Theorem, containing the main results of this Section, we analyse the semimartingale $\tilde{X}^{(T)} - H - [\tilde{X}^{(T)}, H] $ in the sense that we want to consider separately local martingale parts and finite variation parts. 
We have  the following decompositions for our finite variation processes into a absolutely continuous part, a singular continuous part and a jump part.
The process $A$ admits the representation
\begin{equation}\label{eq:decA}
	A_t = \int_0^t A^{\textup{ac}}_s ds + A^{\textup{sing}}_t + \sum_{0 <s\leq t}\Delta A_s, \quad \text{ for all }t \geq 0,
\end{equation}
where $(A^{\textup{ac}}_t)_{t \geq 0}$ is a non-negative predictable process such that $\int_0^t|A^{\textup{ac}}_s|ds<\infty$ and $(A^{\textup{sing}}_t)_{t \geq 0}$ is an increasing and continuous process with $A^{\textup{sing}}_0 = 0$ such that $dA^{\textup{sing}}_s (\omega) \perp ds $, for almost all $\omega \in \Omega$.
For the $d$-dimensional semimartingale we find the decomposition $X = X^{\textup{c} }+ X^{\textup{d}} + \tilde A$, where $ X^\textup{c}$ is the continuous local martingale part, $ X^\textup{d}$ a purely discontinuous local martingale, and $\tilde A$ a process of finite variation.  The process $\tilde A$ admits the following representation
\begin{equation}\label{eq:dectildeA}
	\tilde A_t = \int_0^t \tilde A^{\textup{ac}}_s ds + \tilde A^{\textup{sing}}_t + \sum_{0 <s\leq t}\Delta \tilde A_s, \quad \text{ for all } t \geq 0,
\end{equation}
where $(\tilde A^{\textup{ac}}_t)_{t \geq 0}$ is a non-negative predictable process such that $\int_0^t|\tilde A^{\textup{ac}}_s|ds<\infty$ and $(\tilde A^{\textup{sing}}_t)_{t \geq 0}$ is an increasing and continuous process with $\tilde A^{\textup{sing}}_0 = 0$ such that $d\tilde A^{\textup{sing}}_s (\omega) \perp ds $, for almost all $\omega \in \Omega$.
Moreover, since $\langle X^\textup{c} \rangle$ is of finite variation and continuous, %
$\langle X^\textup{c}\rangle$ can be written as
\begin{equation}\label{eq:decXc}
	\langle X^\textup{c}\rangle_t = \int_0^t \psi_s ds + \zeta_t, \text{ for all } t \geq 0,
\end{equation}
with a non-negative predictable process $(\psi_t)_{t \geq 0}$ such that $\int_0^t | \psi_s| ds < \infty$ and an increasing and continuous process $(\zeta_t)_{t \geq 0}$ with $\zeta_0 = 0$ such that $d\zeta_s(\omega) \perp ds$, for almost all $\omega \in \Omega$.
Furthermore,  \citet[Lemma 2.5]{fontana2018general} yields that the predictable and increasing process  $\bar{\mu}^{(T)}$ admits the unique decomposition 
\begin{equation}\label{eq:decmubar}
\bar{\mu}_t^{(T)} = \int_0^t m_s ds + \nu_t + \sum_{0<s\leq t} \Delta \bar{\mu}_s^{(T)},  \text{ for all } t \in [0,T], \quad T \in \mathbb{R}_+,
\end{equation}
where $(m_t)_{0 \leq t \leq T}$ is a non-negative and predictable process satisfying $\int_0^T m_s ds < +\infty$ a.\,s. and $(\nu_t)_{0 \leq t \leq T}$ is an increasing and continuous process with $\nu_0 =0$ such that $d\nu_s(\omega)\perp ds$, for almost all $\omega\in \Omega$.

Moreover, we introduce for each $T \in \mathbb{R}_+$ the process $Y^{(T)} = (Y_t^{(T)})_{0 \leq t \leq T}$ defined as 
\begin{equation}\label{processY}
Y_t^{(T)} : = \int_0^t \int_0^T g(s,u) \mu(ds,du) + \int_0^t \bar{A}(s,T) dA_s+ \int_0^t \bar{B}(s,T)dX_s.
\end{equation}
For the benefit of a more compact notation we define the following functions $W^{(1)}: \Omega \times \mathbb{R}_+ \times \mathbb{R} \rightarrow \mathbb{R}$, $W^{(2)}: \Omega \times \mathbb{R}_+ \times \mathbb{R}\times \{0,1\} \rightarrow \mathbb{R}$ as
\begin{align}\label{functionW}
	W^{(1)}(\omega;s,y) & := e^{g(\omega;s-,s)\Delta \bar{\mu}_s^{(T)}(\omega)}(e^{-y}-1),\nonumber\\
	W^{(2)}(\omega;s,y,z) & := W^{(1)}(\omega;s,y)z.
\end{align}
The functions $W^{(1)}$ and $W^{(2)}$ are $\mathcal{P}\otimes \mathcal{B}(\mathbb{R})$-measurable and $\mathcal{P}\otimes \mathcal{B}(\mathbb{R}\times \{0,1\})$-measurable, respectively. We assume that $W^{(1)} \in G_{\textup{loc}}(\mu^{Y^{(T)}})$, ensuring that $W^{(1)}\ast\mu^{Y^{(T)}}$ makes sense as an integral with respect to the random measure $\mu^{Y^{(T)}}$ and we have
\begin{align}\label{Weinsmu}
&\bigl(W^{(1)}\ast \mu^{Y^{(T)}}\bigr)_t \nonumber\\
&\qquad= \sum_{0 < s \leq t} e^{g(s-,s)\Delta \bar{\mu}_s^{(T)}} (e^{-\int_s^T g(s,u) \mu(\{s\}\times du)-\bar{A}(s,T)\Delta A_s - \bar{B}(s,T) \Delta X_s}-1),
\end{align}
with the integer-valued random measure $\mu^{Y^{(T)}}$ of the semimartingale $Y^{(T)}$, in the sense of \citet[Proposition II-1.16]{JacodShiryaev}, with compensator $\mu^{p,Y^{(T)}}$.
Moreover, we note that $W^{(2)}\ast \mu^{(Y^{(T)},H)}$ is integrable since $H$ is a single jump process and we have
\begin{align}\label{eq:WYT}
	&\bigl(W^{(2)}\ast \mu^{(Y^{(T)},H)}\bigr)_t \nonumber\\
	&\quad= 
	\sum_{0 < s \leq t} e^{g(s-,s)\Delta \bar{\mu}_s^{(T)}}
	\left(e^{-\int_s^T g(s,u) \mu(\{s\}\times du)-\bar{A}(s,T)\Delta A_s - \bar{B}(s,T)\Delta X_s}-1\right)\Delta H_s,
\end{align}
with the integer-valued random measure $\mu^{(Y^{(T)},H)}$ of the semimartingale $(Y^{(T)},H)$. 

Some calculations show by means of \cite{JacodShiryaev} III-6.22 a), the definition of the process $\tilde X^{(T)}$ in Lemma \ref{LemmaBondPriceasStochExp}, and Equations \eqref{eq:decA} -- \eqref{eq:WYT} that the semimartingale $\tilde{X}^{(T)} - H - [\tilde{X}^{(T)},H]$ admits the decomposition
\begin{align}\label{decompSemimartingale}
\tilde{X}^{(T)}_t - H_t - [\tilde{X}^{(T)},H]_t 
&=
- \int_0^t \bar{B}(s,T) d(X^{\textup{c}}_s + X^{\textup{d}}_s) + \int_0^t f(s,s) ds\nonumber \\
& - \int_0^T f(0,u) du 
- \int_0^t \bar{A}(s,T) A^{\textup{ac}}_s ds \nonumber\\
& -\int_0^t \bar{A}(s,T) dA^{\textup{sing}}_s
- \int_0^t \bar{B}(s,T) \tilde A^{\textup{ac}}_sds\nonumber\\
& -\int_0^t \bar{B}(s,T) d\tilde A^{\textup{sing}}_s
 - \left(\int_0^\cdot \int_{(s,T]}  g(s,u)  \mu(ds,du)\right)_t^\textup{cont}\nonumber\\
&  + \int_0^t g(s-,s) m_s ds + \int_0^t g(s-,s) d\nu_s  \nonumber\\
& + \frac{1}{2}\int_0^t\bigl(\bar{B}(s,T)\bigr)^2 \psi_s ds 
+ \frac{1}{2}\int_0^t\bigl(\bar{B}(s,T)\bigr)^2 d\zeta_s \nonumber\\
&  + \sum_{0<s\leq t}\bigl((e^{g(s-,s) \Delta\bar{\mu}_s^{(T)}}-1) (1 - \Delta H_s)
+ \bar{B}(s,T) \Delta X^{\textup{d}}_s\bigr)\nonumber \\
&  - H_t + \bigl(W^{(1)}\ast \mu^{Y^{(T)}}\bigr)_t - \bigl(W^{(2)}\ast \mu^{(Y^{(T)},H)}\bigr)_t,
\end{align}
where we denote by $\mu^{Y^{(T)}}$ the integer-valued random measure of the semimartingale $Y^{(T)}$ with compensator $\mu^{p,Y^{(T)}}$. In addition $\mu^{(Y^{(T)},H)}$ denotes the integer-valued random measure of the semimartingale $(Y^{(T)},H)$.

With decomposition \eqref{decompSemimartingale} of the semimartingale $\tilde{X}^{(T)} - H - [\tilde{X}^{(T)},H]$ defining the representation of the defaultable bond price $P(t,T)$  as stochastic exponential \eqref{bondPriceStochasticExponential}, we are now in a position to proof our main Theorem \ref{maindriftcondzerorecovery}.

Theorem \ref{maindriftcondzerorecovery} generalizes \citet[Theorem 3.4]{fontana2018general}, considering the present setting in the special case where the forward rates are driven by continuous It\^{o}-processes. The following result provides necessary and sufficient conditions for the reference probability measure $\Q$ to be an ELMM for the credit risky financial market with respect to the numeraire $B= \exp(\int_0^\cdot r_t dt)$.
It turns out that in order to ensure no arbitrage in the case of forward rates driven by general non-continuous semimartingales, there is the need of compensating terms, arising in addition to the classical HJM conditions.

As a preliminary, we recall \citet[Lemma 2.1]{fontana2018general}, stating that the compensator $H^p$ of the default indicator process, admits a unique decomposition
\begin{equation}\label{eq:defcomp}
	H_t^p = \int_0^t h_s ds + \lambda_t + \sum_{0<s\leq t} \Delta H_s^p,
\end{equation}
where $(h_t)_{t \geq 0}$ is a non-negative predictable process such that $\int_0^t h_s ds < +\infty$ almost surely and $(\lambda_t)_{t \geq 0}$ is an increasing continuous process with $\lambda_0 =0$ such that $d\lambda_s(\omega) \perp ds$, for almost all $\omega \in \Omega$.
\begin{theorem}\label{maindriftcondzerorecovery}
	Suppose that Assumptions \ref{ramdommeasure} -- \ref{ass:lackingFubini} hold and consider the zero-recovery case. Let $W^{(1)}$ and $W^{(2)}$ be defined as in \eqref{functionW}. Then $\mathbb{Q}$ is an ELMM for the credit risky financial market with respect to the num\'{e}raire $B = \exp(\int_0^\cdot r_s ds)$, if and only if 
	\begin{equation}\label{eq:intloc}
		 \sum_{0<s\leq t}\bigl((e^{g(s-,s) \Delta\bar{\mu}_s^{(T)}}-1) (1 - \Delta H_s)
		+ \bar{B}(s,T) \Delta X^{\textup{d}}_s\bigr)	 \in \mathcal{A}_{\textup{loc}},
	\end{equation}
	a.s for every $T \in \R_+$ and the following conditions hold almost surely:
	\begin{enumerate}[(i)]
		\item for all $0 \leq t \leq T$ and $T \in \mathbb{R}_+$, it holds that 
		\begin{align*}
			\Delta H_t^p = 1 - e^{-g(t-,t) \Delta\bar{\mu}_t^{(T)}} 
			\bigg(& 1  
			-\Delta \Big( \sum_{0<s\leq \cdot}\bar{B}(s,T) \Delta X^{\textup{d}}_s\Big)^p_t
			- \Delta(W^{(1)}\ast \mu^{p,Y^{(T)}})_t  \nonumber\\
			&
			+ \Delta(W^{(2)}\ast \mu^{p,(Y^{(T)},H)})_t\bigg);
		\end{align*}
		\item for all $T\in \mathbb{R}_+$ and for Lebesgue almost every $t \in [0,T]$, it holds that
				\begin{align*}
		r_t & =  f(t,t) - \bar{A}(t,T) A^{\textup{ac}}_t - \left(\int_0^\cdot \int_{(s,T]} g(s,u)  \mu(ds,du)\right)_t^{\textup{ac}}+ g(t-,t) m_t - h_t\nonumber\\
		&  \quad  - \bar{B}(t,T) \tilde A^{\textup{ac}}_t   + \frac{1}{2}\bigl(\bar{B}(t,T)\bigr)^2 \psi_t
		+ \bigl(W^{(1)}\ast \mu^{p,Y^{(T)}}\bigr)_t^{\textup{ac}} \nonumber\\
		&\quad - \bigl(W^{(2)}\ast \mu^{p,(Y^{(T)},H)}\bigr)_t^{\textup{ac}} +\Big( \sum_{0<s\leq t}\bar{B}(s,T) \Delta X^{\textup{d}}_s\Big)^{p,\textup{ac}};
		\end{align*}
		\item for all $0 \leq t \leq T$, $T \in \mathbb{R}_+$, it holds that
		\begin{align*}
				\lambda_t & =  \int_0^t g(s-,s)d \nu_s - \int_0^t  \bar{A}(s,T) d A^{\textup{sing}}_s  - \int_0^t \bar{B}(s,T) d\tilde A^{\textup{sing}}_s\\
				& \quad - \left(\int_0^\cdot \int_{(s,T]} g(s,u)  \mu(ds,du)\right)_t^{\textup{sing}} 
				+ \frac{1}{2} \int_0^t \bigl( \bar{B}(s,T)\bigr)^2 d \zeta_s  \\
				&\quad + \bigl(W^{(1)}\ast \mu^{p,Y^{(T)}}\bigr)_t^{\textup{sing}}- \bigl(W^{(2)}\ast \mu^{p,(Y^{(T)},H)}\bigr)_t^{\textup{sing}} \\
				&\quad
				+\Big( \sum_{0<s\leq t}\bar{B}(s,T) \Delta X^{\textup{d}}_s\Big)^{p,\textup{sing}}.
		\end{align*}
	\end{enumerate}
\end{theorem}
The first condition is directed at the jump parts of the semimartingale describing discounted bond prices. It describes that the jumps of the default compensator are in a precise connection with the jumps of the process $\bar{\mu}^{(T)}$ and our driving processes $X$ and $A$ of the forward rates. The second condition arises from the absolutely continuous part of the semimartingale describing discounted bond prices and is a generalization of the well known drift condition of the HJM framework. The third condition poses a requirement for the singular continuous part of the semimartingale describing discounted bond prices and states a precise matching condition of the singular part of the default compensator and the singular parts from the remaining processes. Due to the generality of the setting, the conditions of Theorem \ref{maindriftcondzerorecovery} are quite involved.

\begin{proof}
	By definition, $\Q$ is an ELMM with respect to the numeraire $B$ if and only if $P(\cdot,T)/B$ is a $\Q$-local martingale for every $T\in \R_+$. By means of Lemma~\ref{LemmaBondPriceasStochExp}, the definition of the num\'{e}raire, see \eqref{eq:numeraire}, and Yor's formula, see \citet[II-8.19]{JacodShiryaev}, discounted credit risky bond prices can be represented as
	\begin{align}\label{eq:dissum}
	\frac{P(t,T)}{B_t} & =  \frac{\mathcal{E}\bigl(\tilde{X}^{(T)} - H - [\tilde{X}^{(T)},H]\bigr)_t}{\mathcal{E}\bigl(\int_0^\cdot r_s ds\bigr)_t} \nonumber\\
	& =  \mathcal{E}\bigl(\tilde{X}^{(T)} - H - [\tilde{X}^{(T)},H]\bigr)_t\mathcal{E}\left(-\int_0^\cdot r_s ds\right)_t\nonumber\\
	& =  \mathcal{E}\left(\tilde{X}^{(T)} - H - [\tilde{X}^{(T)},H]-\int_0^\cdot r_s ds\right)_t.
	\end{align}
 The process $(P(t,T)/B_t)_{0 \leq t \leq T}$ is a $\Q$-local martingale, for every $T \in \mathbb{R}_+$, if and only if the predictable finite variation term of the semimartingale $\tilde{X}^{(T)} - H - [\tilde{X}^{(T)},H]-\int_0^\cdot r_s ds$ vanish. %
 
 		First, we assume that $P(\cdot,T)/B$ is a $\Q$-local martingale, for every $T \in \R_+$. By \citet[I-3,11]{JacodShiryaev} the finite variation part of \eqref{eq:dissum} is of locally integrable variation.
	 By compensating the process $H$ and keeping in mind the decomposition of the compensator of the default indicator process \eqref{eq:defcomp} we obtain 
	\begin{align*}
		\sum_{0<s\leq t} (e^{g(s-,s)\Delta\bar{\mu}_s^{(T)}}-1)\Delta H_s 
		& =  \int_0^t (e^{g(s-,s)\Delta\bar{\mu}_s^{(T)}} -1) dH_s \\
		& =  (\text{local martingale})_t + \int_0^t (e^{g(s-,s)\Delta\bar{\mu}_s^{(T)}} -1) dH_s^p\\
		& =  (\text{local martingale})_t + \sum_{0< s \leq t} (e^{g(s-,s)\Delta\bar{\mu}_s^{(T)}} -1) \Delta H_s^p.
	\end{align*}

	As can now be deduced from the representation \eqref{decompSemimartingale} by compensating $H$, $W^{(1)}\ast \mu^{Y^{(T)}}$ and  $W^{(2)}\ast \mu^{(Y^{(T)},H)}$, it holds that
	\begin{align}\label{eq:discpricerepresentation}
	&\tilde{X}^{(T)}_t - H_t - [\tilde{X}^{(T)},H]_t-\int_0^t r_s ds 
	 = 	M(t,T)\nonumber\\
	& \quad -\int_0^t r_s ds + \int_0^t f(s,s) ds 
	 - \int_0^t \bar{A}(s,T) A^{\textup{ac}}_s ds - \int_0^t \bar{A}(s,T) d A^{\textup{sing}}_s \nonumber\\
	 &\quad - \int_0^t \bar{B}(s,T) \tilde A^{\textup{ac}}_sds - \int_0^t \bar{B}(s,T) d\tilde A^{\textup{sing}}_s\nonumber\\
	& \quad - \left(\int_0^\cdot \int_{(s,T]} g(s,u)  \mu(ds,du)\right)_t^\textup{cont}
	 + \int_0^t g(s-,s) m_s ds + \int_0^t g(s-,s) d\nu_s\nonumber\\
	& \quad + \frac{1}{2}\int_0^t\bigl(\bar{B}(s,T)\bigr)^2 \psi_s ds + \frac{1}{2}\int_0^t\bigl(\bar{B}(s,T)\bigr)^2 d\zeta_s \nonumber \\
	& \quad - \int_0^t h_s ds -\lambda_t
	 + \sum_{0<s\leq t} \left[\left(e^{g(s-,s) \Delta\bar{\mu}_s^{(T)}} -1\right)(1-\Delta H_s^p) - \Delta H_s^p\right]\nonumber\\
	 &\quad +\Big( \sum_{0<s\leq t}\bar{B}(s,T) \Delta X^{\textup{d}}_s\Big)^p
	 + (W^{(1)}\ast \mu^{p,Y^{(T)}})_t - (W^{(2)}\ast \mu^{p,(Y^{(T)},H)})_t,
	\end{align}
	where $M(\cdot,T)$ denotes the local martingale
	\begin{align*}
		M(t,T) &=  - \int_0^T f(0,u) du  - \int_0^t \bar{B}(s,T) d(X^{\textup{c}}_s + X^{\textup{d}}_s)  -(H_t - H^p_t) \\
		& \quad + \int_0^t (e^{g(s-,s)\Delta\bar{\mu}_s^{(T)}} -1) dH_s- \int_0^t (e^{g(s-,s)\Delta\bar{\mu}_s^{(T)}} -1) dH_s^p\\
		& \quad  + \sum_{0<s\leq t}\bar{B}(s,T) \Delta X^{\textup{d}}_s -  \Big( \sum_{0<s\leq t}\bar{B}(s,T) \Delta X^{\textup{d}}_s\Big)^p\\
		&\quad +\bigl(W^{(1)}\ast \mu^{Y^{(T)}}\bigr)_t  - \bigl(W^{(1)}\ast \mu^{p,Y^{(T)}}\bigr)_t \\
		&\quad - \bigl(W^{(2)}\ast \mu^{(Y^{(T)},H)}\bigr)_t + \bigl(W^{(2)}\ast \mu^{p,(Y^{(T)},H)}\bigr)_t.
	\end{align*}
 Being a local martingale for every $T \in \R_+$, implies that  the predictable finite variation part in \eqref{eq:dissum} vanishes, see \citet[I-3.16]{JacodShiryaev}. We investigate separately the absolutely continuous, singular, and jump parts of the predictable finite variation part of \eqref{eq:discpricerepresentation}. We start with the first step to analyse the jump parts and it holds that
		\begin{align}\label{eq:jumppart}
		0 &= \left(e^{g(t-,t) \Delta\bar{\mu}_t^{(T)}} -1\right)(1-\Delta H_t^p)  -\Delta H_t^p + \Delta\bigl(W^{(1)}\ast \mu^{p,Y^{(T)}}\bigr)_t  \nonumber\\
		&\quad+\Delta \Big( \sum_{0<s\leq \cdot}\bar{B}(s,T) \Delta X^{\textup{d}}_s\Big)^p_t- \Delta\bigl(W^{(2)}\ast \mu^{p,(Y^{(T)},H)}\bigr)_t,
		\end{align}
		which corresponds to condition (i).
		
	We continue with the second step and investigate the continuous singular part of the finite variation terms appearing in \eqref{eq:discpricerepresentation} and it follows that
	\begin{align*}
		0 &= - \int_0^t  \bar{A}(s,T) d A^{\textup{sing}}_s  - \int_0^t \bar{B}(s,T) d\tilde A^{\textup{sing}}_s- \left(\int_0^\cdot \int_{(s,T]} g(s,u)  \mu(ds,du)\right)_t^{\textup{sing}} \\
		& \quad + \int_0^t g(s-,s)d \nu_s 
		    + \frac{1}{2} \int_0^t \bigl( \bar{B}(s,T)\bigr)^2 d \zeta_s + \bigl(W^{(1)}\ast \mu^{p,Y^{(T)}}\bigr)_t^{\textup{sing}} \\
		    &\quad - \bigl(W^{(2)}\ast \mu^{p,(Y^{(T)},H)}\bigr)_t^{\textup{sing}} 
		     +\Big( \sum_{0<s\leq t}\bar{B}(s,T) \Delta X^{\textup{d}}_s\Big)^{p,\textup{sing}}- \lambda_t,
	\end{align*}
	for all $0 \leq t \leq T$, $T\in\mathbb{R}$, which yields condition (iii).\\
	As a third step, we consider the densities of the absolutely continuous part of the finite variation terms from \eqref{eq:discpricerepresentation}.
	For all $0 \leq t \leq T$, $T \in \mathbb{R}_+$ it must hold that
		\begin{align} \label{eq:acpart}
			r_t & =  f(t,t) - \bar{A}(t,T) A^{\textup{ac}}_t - \left(\int_0^\cdot \int_{(s,T]} g(s,u)  \mu(ds,du)\right)_t^{\textup{ac}}+ g(t-,t) m_t - h_t\nonumber\\
			&  \quad  - \bar{B}(t,T) \tilde A^{\textup{ac}}_t   + \frac{1}{2}\bigl(\bar{B}(t,T)\bigr)^2 \psi_t
			+ \bigl(W^{(1)}\ast \mu^{p,Y^{(T)}}\bigr)_t^{\textup{ac}} - \bigl(W^{(2)}\ast \mu^{p,(Y^{(T)},H)}\bigr)_t^{\textup{ac}}\nonumber\\
			&\quad  +\Big( \sum_{0<s\leq t}\bar{B}(s,T) \Delta X^{\textup{d}}_s\Big)^{p,\textup{ac}},
		\end{align}
		yielding condition (ii).
		
	Conversely, if the integrability condition \eqref{eq:intloc} is fulfilled, we obtain the locally integrable variation of the finite variation processes in \eqref{decompSemimartingale}. Therefore we are able to compensate the finite variation processes in \eqref{decompSemimartingale} and obtain representation \eqref{eq:discpricerepresentation}. By means of condition (i) -- (iii), the local martingale property of $P(\cdot,T)/B$, for every $T\in\R_+$ directly follows.
\end{proof}

\section{Defaultable term structure modeling with general recovery schemes} \label{sec:cRrecovery}
In this section we study the case with general recovery. As already explained in the introduction, there are many possible specifications
and modelling approaches for recovery, and in particular a high uncertainty exists in this part. We therefore pose only minimal assumptions on the recovery scheme, and
assume that  $\xi$, is $\mathbb{F}$-adapted, \cadlag \ decreasing, non-negative with $\xi_0 =1$. 
Then there exists a decreasing \cadlag \ process $R$ with bounded jump size $-1 \leq \Delta R \leq 0$, such that 
$$\xi = \mathcal{E}(R). $$ The process $R$ is \cadlag \ and has bounded jumps, hence it is locally bounded and therefore special.

We denote by $$ \tau:=\inf\{t \in [0,T]: \xi_t = 0\}$$ the time point where our recovery process reaches 0, corresponding to the time point where the credit risky bond becomes totally worthless. Note that there could be many defaults before this point is reached. From a technical perspective,  before $\tau$, the recovery process is able to control the amount of loss of the value of the bond price at several credit events. After $\tau$, since then the multiplying factor $\xi$ vanishes and hence bond prices as well, all other remaining parameters may vary freely, thus no no conditions on forward rates are needed after $\tau$.

We note that the continuous local martingale part $R^{\textup{c}}$ of the finite variation process $R$ is  zero. Hence, the canonical decomposition of the process $R$ is %
\begin{equation*}
	R_t = \left(x \ast \bigl(\mu^R - \mu^{p,R}\bigr)\right)_t - C_t, \quad \text{for all } 0 \leq t \leq T, T\in \mathbb{\R}_+,
\end{equation*}
where $\mu^R$ denotes the integer-valued random jump measure of $R$, $\mu^{p,R}$ its compensator, and $(C_t)_{0 \leq t \leq T}$, $T\in \mathbb{\R}_+$ is an increasing predictable process such that $\Delta C_t = - \int_{[-1,0]}x \mu^{p,R}(\{t\},dx)$, for all $t \in [0,T]$, $T\in \mathbb{\R}_+$.

The starting point in this section is therefore the following term-structure of credit risky bonds with general recovery, 
\begin{equation}\label{eq:prices_rec}
P(t,T) = \mathcal{E}(R)_t \exp \left(  -\int_t^Tf(t,u)du-\int_{(t,T]}g(t,u)\mu_t(du) \right), \qquad 0 \le t \le T .
\end{equation}
Note that this representation is a generalization of the term structure with zero recovery. %
Recall the function $W^{(1)}: \Omega \times \mathbb{R}_+ \times \mathbb{R} \rightarrow \mathbb{R}$ from \eqref{functionW}  and define the function $W^{(3)}: \Omega \times \mathbb{R}_+ \times \mathbb{R}\times [-1,0] \rightarrow \mathbb{R}$ as
\begin{align}\label{functionW3}
W^{(3)}(\omega;s,y,x) & := W^{(1)}(\omega;s,y)x,
\end{align}
where $W^{(3)}$ is $\mathcal{P}\otimes \mathcal{B}(\R \times [-1,0])$-measurable. 

\subsection{Absence of arbitrage}
As already mentioned, arbitrage-free credit risky markets can be characterized in terms of equivalent local martingale measures (ELMM). In this section, we seek necessary and sufficient conditions for the reference probability measure $\mathbb{Q}$ to  be an ELMM .

We recall the process $Y^{(T)}$ defined in \eqref{processY} and  denote by $\mu^{(Y^{(T)},R)}$ the random jump measure associated to the two-dimensional semimartingale $(Y^{(T)}, R)$ with its compensator $\mu^{p,(Y^{(T)},R)}$. Moreover, we decompose the process $C$ as
\begin{equation}\label{eq:decC}
C_t = \int_0^t C^{\textup{ac}}_s ds + C^{\textup{sing}}_t + \sum_{0 <s\leq t}\Delta C_s, \quad \text{ for all }t \geq 0,
\end{equation}
where $(C^{\textup{ac}}_t)_{t \geq 0}$ is a non-negative predictable process such that $\int_0^t|C^{\textup{ac}}_s|ds<\infty$ and $(C^{\textup{sing}}_t)_{t \geq 0}$ is an increasing and continuous process with $C^{\textup{sing}}_0 = 0$ such that $dC^{\textup{sing}}_s (\omega) \perp ds $, for almost all $\omega \in \Omega$.

The following theorem  states the desired conditions to ensure the absence of arbitrage. It  generalizes \citet[Theorem 3.12]{fontana2018general} where only continuous forward rates $f$ and $g$ were considered. 
\begin{theorem}\label{maindriftcondgeneralrecovery}
	Suppose that Assumptions \ref{ramdommeasure} -- \ref{ass:lackingFubini}  hold. %
	Then the probability measure $\mathbb{Q}$ is an ELMM, %
	if and only if 	
	\begin{equation}\label{eq:intlocrec}
	\sum_{0<s\leq t}\bigl((e^{g(s-,s) \Delta\bar{\mu}_s^{(T)}}-1) (1 + \Delta R_s)
	+ \bar{B}(s,T) \Delta X^{\textup{d}}_s\bigr)	 \in \mathcal{A}_{\textup{loc}},
	\end{equation}
	a.s for every $T \in \R_+$ and the following conditions hold almost surely:
	\begin{enumerate}[(i)]
		\item for all $T\in \mathbb{R}_+$ and for Lebesgue almost every $t \in [0,T]$, it holds that
			\begin{align*} %
		r_t &=  f(t,t) - \bar{A}(t,T) A^{\textup{ac}}_t  	- \bar{B}(t,T) \tilde A^{\textup{ac}}_t 
		- \left(\int_0^\cdot \int_{(s,T]} g(s,u)  \mu(ds,du)\right)^{\textup{ac}}_t  \nonumber\\
		&  \quad  +  g(t-,t) m_t + \frac{1}{2}\bigl(\bar{B}(t,T) \bigr)^2 \psi_t + \Big( \sum_{0<s\leq \cdot}\bar{B}(s,T) \Delta X^{\textup{d}}_s\Big)^{p,\textup{ac}}_t \\
		&\quad +  \bigl(W^{(1)} \ast \mu^{p,Y^{(T)}}\bigr)^{\textup{ac}}_t  + \bigl(W^{(3)} \ast \mu^{p,(Y^{(T)},R)}\bigr)^{\textup{ac}}_t -C^{\textup{ac}}_t;
		\end{align*}
		\item for all $0 \leq t \leq T$ and $T\in \mathbb{R}_+$, it holds that
		\begin{align*}
			\Delta C_t  = 1- e^{-g(t-,t) \Delta\bar{\mu}_t^{(T)}}
			\bigg(1&-\Delta \bigl(W^{(1)} \ast \mu^{p,Y^{(T)}}\bigr)_t  - \Delta \bigl(W^{(3)} \ast \mu^{p,(Y^{(T)},R)}\bigr)_t \\
			&-\Delta \Big( \sum_{0<s\leq \cdot}\bar{B}(s,T) \Delta X^{\textup{d}}_s\Big)^p_t\bigg);
		\end{align*}
		\item for all $0 \leq t \leq T$, $T \in \mathbb{R}_+$, it holds that
		\begin{align*} 
		C^{\textup{sing}}_t &= - \int_0^t \bar{A}(s,T) dA^{\textup{sing}}_s - \int_0^t \bar{B}(s,T) d\tilde A^{\textup{sing}}_s\\
		&\quad - \left(\int_0^t \int_{(s,T]} g(s,u)  \mu(ds,du)\right)^{\textup{sing}}  + \int_0^t g(s-,s) d\nu_s \\
		&\quad+ \frac{1}{2}\int_0^t\bigl(\bar{B}(s,T) \bigr)^2 d\zeta_s
		+ \Big( \sum_{0<s\leq \cdot}\bar{B}(s,T) \Delta X^{\textup{d}}_s\Big)^{p,\textup{sing}}_t \\
		&\quad+  \bigl(W^{(1)} \ast \mu^{p,Y^{(T)}}\bigr)_t^{\textup{sing}}  + \bigl(W^{(3)} \ast \mu^{p,(Y^{(T)},R)}\bigr)_t^{\textup{sing}}.
		\end{align*}
	\end{enumerate}
\end{theorem}
The interpretation of the conditions is similar as in Theorem \ref{maindriftcondzerorecovery}. Condition (i), (ii), and (iii) arise from the absolutely continuous part, the jump part, and the singular part of the semimartingale $\tilde{X}^{(T)}_t + R_t + [\tilde{X}^{(T)},R]_t -\int_0^t r_s ds$, describing the discounted credit risky bond process, respectively. Due to the general setting the conditions are complex, but reveal that a precise relationship between the underlying processes needs to be satisfied. 

It might be interesting to remark, that regarding practical application of the result, discontinuities are well-acknowledged in the literature, see \cite{GehmlichSchmidt2016MF} for references. Often, chosing $A$ deterministic will be sufficient to incorporate jumps visible already at the current time. This is a necessary condition for affine semimartingale models as shown in \cite{KellerResselSchmidtWardenga2019}. Also in \cite{fontana2020term}, the multiple-yield curve market is analyzed under this additional assumption. 

\begin{proof}
By means of \eqref{eq:prices_rec} and Lemma \ref{eq:reprcRBond} it is possible to represent credit risky bond prices $P(t,T)$ by
	\begin{equation*}
		P(t,T) = \mathcal{E}(R)_t\exp\bigl(X_t^{(T)}\bigr),
	\end{equation*}
	with the process $X_t^{(T)}$ defined in \eqref{processXT}. From the proof of Lemma \ref{LemmaBondPriceasStochExp} we know that $\exp(X_t^{(T)}) = \mathcal{E}(\tilde{X}^{(T)}_t)$, where $\tilde{X}^{(T)}_t$ is defined as in Lemma \ref{LemmaBondPriceasStochExp}. By means of Yor's formula, see \citet[II-8.19]{JacodShiryaev}, credit risky bond prices admit a representation as a stochastic exponential
	\begin{equation*}
		P(t,T) = \mathcal{E}\bigl(\tilde{X}^{(T)} + R + [\tilde{X}^{(T)},R]\bigr)_t.
	\end{equation*}
By means of the definition in  \eqref{eq:decA} -- \eqref{Weinsmu} and \eqref{functionW3} we obtain the following semimartingale representation of $\tilde{X}^{(T)}_t + R_t + [\tilde{X}^{(T)},R]_t$,
	\begin{align}\label{eq:decSmgenrec}
		\tilde{X}^{(T)}_t + R_t + [\tilde{X}^{(T)},R]_t &=   - \int_0^T f(0,u) du
		- \int_0^t \bar{B}(s,T) d(X^{\textup{c}}_s + X^{\textup{d}}_s) + R_t  \nonumber\\
		&\quad + \int_0^t f(s,s) ds 
		- \int_0^t \bar{A}(s,T) A^{\textup{ac}}_s ds \nonumber\\
		& \quad  -\int_0^t \bar{A}(s,T) dA^{\textup{sing}}_s- \int_0^t \bar{B}(s,T) \tilde A^{\textup{ac}}_sds \nonumber\\
		&\quad- \int_0^t \bar{B}(s,T) d\tilde A^{\textup{sing}}_s- \left(\int_0^\cdot \int_{(s,T]} g(s,u)  \mu(ds,du)\right)_t^\textup{cont}\nonumber\\
		& \quad + \int_0^t g(s-,s) m_s ds + \int_0^t g(s-,s) d\nu_s  \nonumber\\
		& \quad+ \frac{1}{2}\int_0^t\bigl(\bar{B}(s,T)\bigr)^2 \psi_s ds 
		+ \frac{1}{2}\int_0^t\bigl(\bar{B}(s,T)\bigr)^2 d\zeta_s \nonumber\\
		& \quad + \sum_{0<s\leq t}\bigl((e^{g(s-,s) \Delta\bar{\mu}_s^{(T)}}-1) (1 + \Delta R_s)
		+ \bar{B}(s,T) \Delta X^{\textup{d}}_s\bigr) \nonumber\\
		& \quad + \bigl(W^{(1)}\ast \mu^{Y^{(T)}}\bigr)_t + \bigl(W^{(3)}\ast \mu^{(Y^{(T)},R)}\bigr)_t.
	\end{align}

	By definition, $\Q$ is an ELMM with respect to the numeraire $B$, if and only if, $P(\cdot,T)/B$ is a $\Q$-local martingale for every $T\in\R_+$.
	
	First, assume that $P(\cdot,T)/B$ is a $\Q$-local martingale for every $T \in \R_+$. Then, %
	the finite variation part of \eqref{eq:decSmgenrec} is of locally integrable variation. Since we aim at a decomposition of \eqref{eq:decSmgenrec} as a local martingale plus some predictable finite variation processes, we consider the following. 
	
	Moreover, %
	the jumps $\Delta R_t = \int_{[-1,0]} x \mu^R(\{t\},dx)$ can be represented as
	\begin{equation*}
		\Delta R_t = (\textup{local martingale})_t + \int_{[-1,0]} x \mu^{p,R}(\{t\},dx)= (\textup{local martingale})_t -\Delta C_t,
	\end{equation*}
	with $\Delta C_t = - \int_{[-1,0]}x \mu^{p,R}(\{t\},dx)$, for all $t \in [0,T]$, $T\in \mathbb{\R}_+$.
	Therefore, 
	\begin{align*}
		\sum_{0<s \leq t}  \bigl(e^{g(s-,s) \Delta\bar{\mu}_s^{(T)}}  -1\bigr) (1+ \Delta R_s)
		& =  (\textup{local martingale})_t \\
		& \quad+ \sum_{0<s \leq t} \bigl(e^{g(s-,s) \Delta\bar{\mu}_s^{(T)}}  -1\bigr)(1-\Delta C_s).
	\end{align*}

Since $W^{(1)} \in G_{\textup{loc}}(\mu^{Y^{(T)}})$ and $R$ have bounded jumps, we are able to compensate the processes $W^{(1)} \ast \mu^{Y^{(T)}}$ and   $W^{(3)} \ast \mu^{(Y^{(T)},R)}$ with the corresponding compensators $\mu^{p,Y^{(T)}}$ and $\mu^{p,(Y^{(T)},R)}$, respectively.
We finally obtain the following representation by means of \eqref{eq:decSmgenrec}
	\begin{align} \label{eq:decompSemimartingalegeneralrecovery}
		\tilde{X}^{(T)}_t + R_t + [\tilde{X}^{(T)},R]_t &-\int_0^t r_s ds  \nonumber\\
		&= M(t,T)   -\int_0^t r_s ds +   \int_0^t f(s,s) ds
		- \int_0^t \bar{A}(s,T)A^{\textup{ac}}_s ds  \nonumber \\
		&  \quad - \int_0^t \bar{A}(s,T) dA^{\textup{sing}}_s
		- \int_0^t \bar{B}(s,T) \tilde A^{\textup{ac}}_sds \nonumber\\
		&  \quad - \int_0^t \bar{B}(s,T) d\tilde A^{\textup{sing}}_s - \left(\int_0^t \int_{(s,T]} g(s,u)  \mu(ds,du)\right)^\textup{cont}\nonumber\\
		& \quad + \int_0^t g(s-,s) m_s ds + \int_0^t g(s-,s) d\nu_s\nonumber\\
		&\quad  +\frac{1}{2}\int_0^t\bigl(\bar{B}(s,T) \bigr)^2 \psi_s ds+ \frac{1}{2}\int_0^t\bigl(\bar{B}(s,T) \bigr)^2 d\zeta_s\nonumber\\
		& \quad +  \bigl(W^{(1)} \ast \mu^{p,Y^{(T)}}\bigr)_t  + \bigl(W^{(3)} \ast \mu^{p,(Y^{(T)},R)}\bigr)_t\nonumber\\
		& \quad+ \sum_{0<s \leq t} \bigl(e^{g(s-,s) \Delta\bar{\mu}_s^{(T)}}  -1\bigr)(1-\Delta C_s) \nonumber\\
		&\quad+ \Big( \sum_{0<s\leq \cdot}\bar{B}(s,T) \Delta X^{\textup{d}}_s\Big)^p_t -C_t,
	\end{align}
	where $M(\cdot,T)$ denotes the local martingale
	\begin{align*}
		M(t,T) &=  - \int_0^T f(0,u) du
		- \int_0^t \bar{B}(s,T) d(X^{\textup{c}}_s + X^{\textup{d}}_s)\\
		&\quad 	+ 	 \bigl( e^{g(\cdot \,-,\cdot) \Delta\bar{\mu}^{(T)}} \bigr)x \ast (\mu^R  -\mu^{p,R})_t\\
		& \quad  + \sum_{0<s\leq t}\bar{B}(s,T) \Delta X^{\textup{d}}_s -  \Big( \sum_{0<s\leq t}\bar{B}(s,T) \Delta X^{\textup{d}}_s\Big)^p\\
		&\quad +\bigl(W^{(1)}\ast \mu^{Y^{(T)}}\bigr)_t  - \bigl(W^{(1)}\ast \mu^{p,Y^{(T)}}\bigr)_t \\
		&\quad + \bigl(W^{(3)}\ast \mu^{(Y^{(T)},R)}\bigr)_t - \bigl(W^{(3)}\ast \mu^{p,(Y^{(T)},R)}\bigl)_t.
	\end{align*}
	The finite variation part of \eqref{eq:decompSemimartingalegeneralrecovery} must vanish. 
	We consider separately the absolutely continuous, singular, and jump parts of the finite variation parts.
	In a first step we analyze the jump parts and it must hold that
	\begin{align}\label{eq:jumppartgeneralrecovery}
	0 &= \Delta \bigl(W^{(1)} \ast \mu^{p,Y^{(T)}}\bigr)_t  + \Delta \bigl(W^{(3)} \ast \mu^{p,(Y^{(T)},R)}\bigr)_t +\Delta \Big( \sum_{0<s\leq \cdot}\bar{B}(s,T) \Delta X^{\textup{d}}_s\Big)^p_t\nonumber\\
	&\quad+  \bigl(e^{g(t-,t) \Delta\bar{\mu}_t^{(T)}}  -1\bigr)(1-\Delta C_t) -\Delta C_t ,
	\end{align}
	for all $0\leq t \leq T$, $T \in \mathbb{R}_+$, which yields condition (ii). 

	Considering the continuous singular part of the finite variation terms appearing in \eqref{eq:decompSemimartingalegeneralrecovery} it follows that
	\begin{align*} 
		0 = &   - \int_0^t \bar{A}(s,T) dA^{\textup{sing}}_s - \int_0^t \bar{B}(s,T) d\tilde A^{\textup{sing}}_s- \left(\int_0^t \int_{(s,T]} g(s,u)  \mu(ds,du)\right)^{\textup{sing}}  \nonumber\\
		& + \int_0^t g(s-,s) d\nu_s + \frac{1}{2}\int_0^t\bigl(\bar{B}(s,T) \bigr)^2 d\zeta_s
		+ \Big( \sum_{0<s\leq \cdot}\bar{B}(s,T) \Delta X^{\textup{d}}_s\Big)^{p,\textup{sing}}_t \\
		&+  \bigl(W^{(1)} \ast \mu^{p,Y^{(T)}}\bigr)_t^{\textup{sing}}  + \bigl(W^{(3)} \ast \mu^{p,(Y^{(T)},R)}\bigr)_t^{\textup{sing}} -C^{\textup{sing}}_t\nonumber,
	\end{align*}
	for all $0 \leq t \leq T$, $T \in \mathbb{R}_+$, which yields condition (iii).\\
	It remains to consider the absolutely continuous part of the finite variation terms from \eqref{eq:decompSemimartingalegeneralrecovery}. For all $0 \leq t \leq T$, $T\in \mathbb{R}_+$ it must hold that
	\begin{align*} %
		r_t &=  f(t,t) - \bar{A}(t,T) A^{\textup{ac}}_t  	- \bar{B}(t,T) \tilde A^{\textup{ac}}_t 
		- \left(\int_0^\cdot \int_{(s,T]} g(s,u)  \mu(ds,du)\right)^{\textup{ac}}_t  \nonumber\\
		&  \quad  +  g(t-,t) m_t + \frac{1}{2}\bigl(\bar{B}(t,T) \bigr)^2 \psi_t + \Big( \sum_{0<s\leq \cdot}\bar{B}(s,T) \Delta X^{\textup{d}}_s\Big)^{p,\textup{ac}}_t \\
		&\quad +  \bigl(W^{(1)} \ast \mu^{p,Y^{(T)}}\bigr)^{\textup{ac}}_t  + \bigl(W^{(3)} \ast \mu^{p,(Y^{(T)},R)}\bigr)^{\textup{ac}}_t -C^{\textup{ac}}_t,
	\end{align*}
	which corresponds to condition (i).
	
	Conversely, if the integrability condition \eqref{eq:intlocrec} is fulfilled, we obtain that the finite variation processes in \eqref{eq:decSmgenrec} are of locally integrable variation, such that we obtain representation \eqref{eq:decompSemimartingalegeneralrecovery}. The local martingale property of $P(\cdot,T)/B$, for every $T \in \R_+$, follows with condition (i) -- (iii).
\end{proof}

\section{Conclusion}\label{sec:conclusion}
In this work we studied a credit risky market driven by finite-dimensional semimartingales under minimal assumptions. It turned out that using semimartingales as drivers in comparison to L\'evy processes, or semimartingales with absolutely continuous characteristics requires the extension of the HJM setting for taking stochastic discontinuities into account. Here we studied the most general extension where a random measure drives a second integral in the HJM representation. All forward random fields were driven by general semimartingales and we obtained necessary and sufficient drift conditions characterizing local martingale measures. This is the key step in guaranteeing that the market is free of arbitrage in the sense of NAFLVR. 

Incorporating stochastic discontinuities into models for financial markets, and in particular term structure models, has only been taken up recently in the literature, although their presence was acknowledged by practitioners for quite a while (see, for example, \cite{Piazzesi2001}). The present work gives a framework which is able to capture this and builds the foundation for building decisive models which are able to incorporate stochastic discontinuities and predictable components of jumps. 

Possible extensions of the present work are: the extension to infinite-dimensional drivers, compare with the very general setting in \cite{grafendorfer2016infinite}; the extension to  multiname-credit like in \cite{bielecki2014bottom} and \cite{giesecke2013default}; the study of  CDO term structures like in \cite{FilipovicSchmidtOverbeck2011}. 


\end{document}